\documentclass[9pt,shortpaper,twoside,web]{ieeecolor}
\usepackage{generic}
\usepackage{cite}
\usepackage{amsmath,amssymb,amsfonts}
\usepackage{algorithmic}
\usepackage{graphicx}
\usepackage{textcomp}
\usepackage{url}

\usepackage{amsthm}

\def\BibTeX{{\rm B\kern-.05em{\sc i\kern-.025em b}\kern-.08em
    T\kern-.1667em\lower.7ex\hbox{E}\kern-.125emX}}
\theoremstyle{plain}
\newtheorem{theorem}{Theorem}

\newtheorem{proposition}{Proposition}
\theoremstyle{proposition}
\theoremstyle{assumption}
\newtheorem{definition}{Definition}

\newtheorem{remark}{Remark}
\theoremstyle{goal}

\theoremstyle{assumption}
\newtheorem{assumption}{Assumption}

\begin{document}
\title{Estimation of Unknown Parameters in Presence of Perturbations and Noises with Application to GPEBO Design}
\author{Anton Glushchenko, \IEEEmembership{Member, IEEE}, Konstantin Lastochkin
\thanks{A. I. Glushchenko is with V.A. Trapeznikov Institute of Control Sciences RAS, Moscow, Russia (phone: +79102266946; e-mail: aiglush@ipu.ru).}
\thanks{K. A. Lastochkin is with V.A. Trapeznikov Institute of Control Sciences RAS, Moscow, Russia (e-mail: lastconst@ipu.ru).}}

\maketitle

\begin{abstract}
A problem of online estimation of unknown parameters is considered for a linear regression equation, which is affected by an additive perturbation that can be caused by measurement noise (that corrupts regressor and regressand), as well as external perturbations. Known approaches to solve this problem typically have one of the following disadvantages: 1) they ensure convergence of a parametric error to a compact set with non-adjustable bound, 2) independence of all system regressor elements from the perturbation/noise is required to annihilate {them}, 3)~an instrumental variable is needed to be selected. On the basis of the novel perturbation annihilation procedure, in the present paper, we propose three new estimation laws, which are free from the above-mentioned drawbacks and ensure exponential convergence of the parametric error to an arbitrarily small neighborhood of zero, particularly, in case more than a half (not all) of the regressor elements are independent from additive perturbation. One of the proposed estimation laws is used for the design of Generalized Parameter Estimation-Based Observer (GPEBO) for nonlinear affine systems to enhance GPEBO performance in case when the measured system output is corrupted by noise. The theoretical results are supported by examples and mathematical modelling.
\end{abstract}

\section{Introduction}
\label{sec:introduction}
It is well-known that a large number of challenging control problems with the system/environment uncertainties can be reduced to a task of unknown parameters online estimation of a simple Linear Regression Equation (LRE) \cite{1, 2, 3, 4, 5, 6}.

Therefore, in recent years, many novel parameter estimators have been proposed, which enhance the estimation performance. P. Lion and G. Kreisselmeier proposed the Dynamic and Memory Regressor Extension schemes \cite{1, 8}, which relax the persistent excitation condition and allow one to increase the rate of parametric convergence arbitrarily via making the learning gain of the estimation law higher. G. Chowdhary \cite{9} proposed the concurrent learning estimation law, which uses past and current data concurrently to relax the persistence of excitation condition and provides parameters convergence under Finite Excitation (FE) assumption. Then, being inspired by \cite{9}, many researchers proposed novel data-based and integral-like regressor extension schemes \cite{6, 7, 10} to relax the persistent excitation condition without data stacking. S. Aranovskiy et al. \cite{11} proposed the Dynamic Regressor Extension and Mixing (DREM) procedure, which provides element-wise monotonous transient response for each parametric error and relaxes persistence of excitation requirement. R. Marino and P. Tomei proposed in \cite{12} and M. Broucke improved in \cite{13} the so-called  $\mu$-modification, which uses the excitation on the subspaces providing the exponential stability of the well-known gradient descent estimator without persistence of excitation condition. Being inspired by the heavy-ball method from optimization, A. Annaswamy et al. \cite{14} proposed the novel high-order tuners, which ensure accelerated rate of convergence and provide arbitrary increase of the parametric convergence rate via making the learning gain of the estimation law higher. Exhaustive overviews of the results in the direction of performance enhancement of parameters online estimation and persistent excitation condition relaxation are presented in \cite{15, 16, 17}.

 In modern literature, significantly less attention is paid to the task of parameters online estimation in the presence of perturbations and noises. However, considering practical scenarios, external perturbations always come along with the motion of the system, and the measurement noise also always corrupts available measurements of the system state and, therefore, the LRE is always affected by the additive perturbation.

Typically, in the papers related to the unknown parameters online estimation, the robustness with respect to the perturbations and noises has been proved in the sense of Uniform Ultimate Boundedness (UUB) (for example, see Theorem 2 from \cite{6}, Proposition 4 from \cite{7}, Lemma 3.1 from \cite{18} and Theorem 4 from \cite{19}). UUB is a very weak property for practical applications, as the obtained ultimate upper bound is almost always a simple ratio of the perturbation and excitation magnitudes and cannot be reduced arbitrarily by choice of estimation law tuning gains. From the practical point of view, such property means that the designer/user cannot mitigate the effect of the measurement noises and perturbations on the obtained unknown parameters estimate, which, consequently, can be arbitrarily bad, especially if the low order excitation scenario is considered.

To overcome such a drawback, in literature, the estimation laws have been proposed \cite{Polyak, Granichin, 20, 21, 22}, which in the presence of perturbations and measurement noises provide convergence of the parameter estimation error to a) zero, b) a compact set, which can be reduced arbitrarily via increase of some tuning gains of the estimator. Let us briefly consider some existing estimation laws with such properties. {The problem of estimation of LRE parameters with arbitrary (in some sense) noise was considered in \cite{Polyak, Granichin}. The properties of such approaches are proved for LRE with random regressors, perturbations and in the case when some {\emph{independence}} assumptions hold. Particularly, in the stationary and deterministic case the convergence conditions from \cite{Polyak} are violated when the regressor and perturbation contain at least one common frequency, and convergence conditions from \cite{Granichin} are violated if the regressor and perturbation contain common spectral line at zero frequency.} In \cite{20} the power noise filtration method is proposed for the DREM procedure, which provides the power noise mitigation in the scalar regression equations and convergence of the parameter estimation error to a compact set with arbitrarily small bound if the following conditions are met: {\it i}) the noise is independent from the scalar regressor and {\it ii}) the noise magnitude is significantly lower in comparison with the scalar regressor one. In \cite{21}, for the scalar regression equations obtained by the DREM procedure, the estimation law with averaging is proposed, which ensures asymptotic exact estimation of the unknown parameters if the scalar regressor is independent from the external perturbation. In \cite{22}, the requirement of independence of the regressor from the perturbation is relaxed via the novel Instrumental Variable based DREM procedure. However, the estimator from \cite{22} requires to choose the Instrumental Variable properly, and therefore, it is {\it bona fide} applicable only for estimation of linear systems parameters and in case the input of such system is independent from measurement noise and external perturbations. In \cite{23}, some analytical and numerical comparison of the above-mentioned estimation laws, which provide parametric error convergence to zero or to arbitrarily small compact set, is presented. 

In the present paper, we propose novel parameter estimators for online estimation of unknown parameters in the presence of perturbations and noises, which have the following distinguishing features and advantages over the results \cite{Polyak, Granichin, 20, 21, 22}.
\begin{enumerate}
\item[\textbf{DF1.}] If sufficiently large number of regressor elements are independent from the perturbation that affects the LRE, and some excitation requirements are met, then the proposed estimation laws ensure parametric error convergence to some compact set, which size can be arbitrarily reduced via the choice of tuning gain of the estimator. This is a main novelty and benefit of the proposed estimators in comparison with many existing ones.
\item[\textbf{DF2.}] In comparison with the results of \cite{Polyak, Granichin, 20, 21}, to achieve the parametric error convergence to an arbitrarily small set, the proposed estimation laws require not all elements of the regressor to be independent from the perturbation, but only a sufficiently large number of such elements. 
\item[\textbf{DF3.}] In comparison with the results of \cite{22}, the proposed estimation laws do not require the Instrumental Variable selection and, therefore, they are applicable not only to LRE parameterizations of the linear systems, but to the general LRE.
\end{enumerate}

The Generalized Parameter Estimation-Based Observer (GPEBO) is the system state reconstruction method, which solves the observation task via the system initial condition estimation \cite{3, 4}.  Recently, it has been reported \cite{4} that the GPEBO with G+D interlaced estimator of the system initial conditions  ensures globally exponentially stable observation of the system state iff the system is observable. However, in the presence of measurement noise, the GPEBO augmented with G+D estimator provides only the convergence of both initial conditions and state estimation errors to some residual sets with bound, which cannot be reduced arbitrarily. In the present paper, we augment GPEBO with the novel estimation law to achieve exponential convergence of the state estimation error to an arbitrarily small compact set.

The preliminary results of the paper have been presented at ALCOS-2025 \cite{24} and ECC-2025 \cite{25}. In comparison with such conference papers, in this article we give some unpublished theoretical justifications and provide some auxiliary insights into the properties of proposed estimation law. Moreover, in this paper, we present some ideas from \cite{24, 25} in a clearer manner and, for the first time, show the application of the proposed estimation law to the design of GPEBO for nonlinear affine systems. 

The paper has a standard linear structure, but in order to achieve the space limitations of technical note format, all proofs and auxiliary simulation results are postponed to Supplementary Material \cite{27} with an open access option.

\subsection{Preliminaries}

The following notation is used throughout the paper: $t \ge t_{0} \ge 0$ stands for the time, $f: \left[t_0, \infty \right) \mapsto \mathbb{R}^n$ denotes a function of time, \emph{i.e.}, a continuous-time signal, $f\left(t\right)$ is a value of the signal $f$ at the time instant $t$, $\left| \cdot \right|$ is the absolute value, $\left\| \cdot \right\|$ is the suitable norm of $(\cdot)$, {${\rm{j}}$ denotes an imaginary unit,} ${I_{n \times n}}=I_{n}$ is an identity $n \times n$ matrix, ${0_{n \times n}}$ is a zero $n \times n$ matrix, $0_{n}$ stands for a zero vector of length $n$, ${\rm{det}}\{\cdot\}$ stands for a matrix determinant, ${\rm{adj}}\{\cdot\}$ represents an adjoint matrix, for a symmetric matrix $A\in\mathbb{R}^{n\times n}$ we write $A>0$ if all eigenvalues of $A$ are strictly greater then zero.

Next, we provide a listing of auxiliary definitions and propositions \cite{27}, which will be used axiomatically throughout the paper. {In the below given definitions and propositions, we assume that {\emph{i}}) all integrals and limits exist for all $t \in \left[ {{t_0}{\rm{,\;}}\infty } \right)$ and $h\in\mathbb{R}$, {\emph{ii}}) all signals are extended by zero value for $t\in\left(-\infty, t_0\right)$.}

\begin{definition}
    A signal $\varphi {\rm{:\;}}\left[ {{t_0}{\rm{,\;}}\infty } \right) \mapsto {{\mathbb{R}} ^n}$ is called stationary $\left( {\varphi  \in {\rm{ST}}} \right)$ if there exists $\Lambda=\Lambda^{\top}\ge0$ such that:
\begin{equation*}{
\mathop {{\rm{lim}}}\limits_{T \to \infty } \mathop {{\rm{sup}}}\limits_{t \ge t_0} \left\| {\frac{1}{T}\int\limits_t^{t + T} {\varphi \left( \tau  \right){\varphi ^ \top }\left( \tau  \right)} d\tau  - \Lambda } \right\| = 0.}
\end{equation*}
\end{definition}

\begin{definition} A regressor $\varphi {\rm{:\;}}\left[ {{t_0}{\rm{,\;}}\infty } \right) \mapsto {{\mathbb{R}} ^n}$ is persistently exciting $\left( {\varphi  \in {\rm{PE}}} \right)$ if there exist $\overline{\alpha}  \ge \underline{\alpha} > 0{\rm{,\;}}T_{\rm{PE}} > 0$ such that the following inequalities hold for all $t \ge {t_0}$:
\begin{equation*}
\overline{\alpha} {I_n} \ge {\frac{1}{T_{\rm{PE}}}}\int\limits_t^{t + T_{\rm{PE}}} {\varphi \left( \tau  \right){\varphi ^{\top}}\left( \tau  \right)} d\tau  \ge \underline{\alpha}{I_n}.
\end{equation*}
\end{definition}

\begin{definition} Signals $f{\rm{:\;}}\left[ {{t_0}{\rm{,\;}}\infty } \right) \mapsto {\mathbb{R}^n}$, $g{\rm{:\;}}\left[ {{t_0}{\rm{,\;}}\infty } \right) \mapsto {\mathbb{R}^m}$ are called independent $\left( {f \bot g} \right)$ if for all {$h\in\mathbb{R}$} it holds that: 
\begin{equation*}{
\mathop {{\rm{lim}}}\limits_{T \to \infty } \mathop {{\rm{sup}}}\limits_{t \geqslant {t_0}} \left\| {\frac{1}{T}\int\limits_t^{t + T} {f\left( \tau  \right){g^ \top }\left( {\tau  + h} \right)} d\tau} \right\|=0.}
\end{equation*}
\end{definition}

\begin{definition} Signals $f{\rm{:\;}}\left[ {{t_0}{\rm{,\;}}\infty } \right) \mapsto {\mathbb{R}^n}$, $g{\rm{:\;}}\left[ {{t_0}{\rm{,\;}}\infty } \right) \mapsto {\mathbb{R}^m}$ are called dependent $\left( {f\parallel g} \right)$ if there exists {a continuous function ${C_W} {\rm{:\;}}\mathbb{R} \mapsto  {\mathbb{R}^{n \times m}}$ and $h^*\!\in\!\mathbb{R}$ such that $C_W\left(h^*\right)\ne 0_{n\times m}$} and for all {$h\in\mathbb{R}$} it holds that: 
\begin{equation*}{
\mathop {{\rm{lim}}}\limits_{T \to \infty } \mathop {{\rm{sup}}}\limits_{t \geqslant {t_0}} \left\| {\frac{1}{T}\int\limits_t^{t + T} {f\left( \tau  \right){g^ \top }\left( {\tau  + h} \right)} d\tau  - {C_W}\left( h \right)} \right\|=0.}
\end{equation*}
\end{definition}

 \begin{proposition}
{If $\varphi  \in {\rm{ST}}$ {and $f \bot g \left(f\parallel g\right)$}, then it holds that}
\begin{equation*}{
\begin{gathered}
\mathop {{\rm{lim}}}\limits_{T \to \infty } \mathop {{\rm{sup}}}\limits_{t \geqslant {t_0 + T}} \left\| {\frac{1}{T}\int\limits_{t-T}^{t} {\varphi \left( \tau  \right){\varphi ^ \top }\left( \tau  \right)} d\tau  - \Lambda } \right\| = 0,\\
\mathop {{\rm{lim}}}\limits_{T \to \infty } \mathop {{\rm{sup}}}\limits_{t \geqslant {t_0 + T}} \left\| {\frac{1}{T}\int\limits_{t-T}^{t} {f\left( \tau  \right){g^ \top }\left( {\tau  + h} \right)} d\tau  - {C_W}\left( h \right)} \right\|=0,
\end{gathered}}
\end{equation*}
{where, for $f \bot g$, it is assumed that $C_W\left(h\right)=0$ for all $h\in\mathbb{R}$.}
\end{proposition}

\begin{proposition} {If $\varphi \in {\rm{ST}} \cap {\rm{PE}}$, then $\Lambda  > 0$.}
\end{proposition}

\begin{proposition} Let $f{\rm{:\;}}\left[ {{t_0}{\rm{,\;}}\infty } \right) \mapsto \mathbb{R}$ and $g{\rm{:\;}}\left[ {{t_0}{\rm{,\;}}\infty } \right) \mapsto \mathbb{R}$ be defined as follows {(here $A_i^f{\text{,\;}}A_i^g > 0$ and $\omega _i^f \ne \omega _j^f{\text{, }}\omega _i^g \ne \omega _j^g$ for all $i \ne j$)}:
\begin{equation*}
f\left( t \right) = \sum\limits_{i = 1}^N {A_i^f{\rm{sin}}\left( {\omega _i^ft + \phi _i^f} \right)} {\rm{,\;}}g\left( t \right) = \sum\limits_{i = 1}^N {A_i^g{\rm{sin}}\Bigl( {\omega _i^gt + \phi _i^g} \Bigr)} {\rm{,}}
\end{equation*}
then, if  $\omega _i^f \ne \omega _j^g$ for all $i{\rm{,\;}}j \in \left\{ {1, \ldots {\rm{,}}N} \right\}$ {$\left(\omega _i^f = \omega _j^g\right.$ for some $i{\rm{,\;}}j \in \left\{ {1, \ldots {\rm{,}}N} \right\}\Bigr)$, then:}
\begin{enumerate}
\item[1)]$f \bot g \left( {f\parallel g} \right){\rm{,}}$ 
\item[2)]${f_f} \bot g \left( {{f_f}\parallel g} \right){\rm{,}}$
\item[3)]${f_f} \bot {g_f} \left( {{f_f}\parallel {g_f}} \right){\rm{,}}$ 
\end{enumerate}
where ${f_f}\left( t \right){\rm{:}} = W\left( s \right)\left[ f \right]\left( t  \right){\rm{,\;}}{g_f}\left( t \right){\rm{:}} = W\left( s \right)\left[ g \right]\left( t  \right)$ and $W\left( s \right)$ {is a stable minimum-phase transfer function.}
\end{proposition}

{We also introduce a more general version of proposition 3.}

\begin{proposition}
{{If $f{\rm{:\;}}\left[ {{t_0}{\rm{,\;}}\infty } \right) \mapsto \mathbb{R}$ and $g{\rm{:\;}}\left[ {{t_0}{\rm{,\;}}\infty } \right) \mapsto \mathbb{R}$ are such that a) $f \bot g\left( {f\parallel g} \right){\rm{,}}$ and {b) $f \in L_\infty$ and $g \in L_\infty$,} then}}
\begin{enumerate}
\item[1)]${f_f} \bot g\left( {{f_f}\parallel g} \right){\rm{,}}$
\item[2)]${f_f} \bot {g_f}\left( {{f_f}\parallel {g_f}} \right){\rm{,}}$ 
\end{enumerate}
{{where ${f_f}\left( t \right){\rm{:}} = W\left( s \right)\left[ f \right]\left( t  \right){\rm{,\; }}{g_f}\left( t \right){\rm{:}} = W\left( s \right)\left[ g \right]\left( t  \right)$ and $W\left( s \right)$ is a stable {strictly-proper} minimum-phase transfer function.}}
\end{proposition}

\section{Problem Statement}
We consider a linear regression equation
\begin{equation}\label{1}
z\left( t \right) = {\varphi ^{\top}}\left( t \right)\theta 
\end{equation}
in conjunction with the measurement (observation) equations
\begin{equation}\label{2}
\begin{array}{l}
\hat z\left( t \right){\rm{:}} = z\left( t \right) + {\eta _f}\left( t \right) + {\eta _z}\left( t \right){\rm{,}}\\
\hat \varphi \left( t \right){\rm{:}} = \varphi \left( t \right) + {\eta _\varphi }\left( t \right){\rm{,}}
\end{array}
\end{equation}	
where $\hat z\left( t \right) \!\in\! \mathbb{R}$ and $\hat \varphi \left( t \right) \!\in\! {\mathbb{R}^n}$ are measurable regressand and regressor, $\theta  \in {\mathbb{R}^n}$ stands for unknown parameters, ${\eta _\varphi }\left( t \right)\!\in\!{{\mathbb{R}} ^n}{\rm{,\;}}{\eta _z}\left( t \right) \in {\mathbb{R}}$ denote measurement noises, which are assumed to be bounded deterministic signals, ${\eta _f}\left( t \right) \in \mathbb{R}$ is an unmeasurable external perturbation that is also assumed to be a bounded deterministic signal. 

Substitution of \eqref{2} into \eqref{1} yields a regression equation with both measurable regressor and regressand:
\begin{equation}\label{3}
\hat z\left( t \right) = {\hat \varphi ^{\top}}\left( t \right)\theta  + w\left( t \right){\rm{,}}
\end{equation}
and a bounded additive perturbation $w\left( t \right) \in \mathbb{R}$ that is defined as follows:
\begin{equation}\label{4}
w\left( t \right){\rm{:}} = {\eta _f}\left( t \right) + {\eta _z}\left( t \right) - \eta _\varphi ^{\top}\left( t \right)\theta.
\end{equation}

The goal is to design an estimation algorithm, which, for any \linebreak $\varepsilon  > 0$, ensures existence of $\delta  > 0$ such that it holds for all $T > \delta$ and $i = 1, \ldots {\rm{,}}n$ that: 
\begin{equation}\label{5}
{\tilde \theta _i}\left( t \right) \in {\cal E}{\rm{:}} = \left\{ {{{\tilde \theta }_i} \in \mathbb{R}{\rm{:\;}}\left| {{{\tilde \theta }_i}} \right| \le \varepsilon } \right\}{\rm{\;as\;}}t \to \infty {\rm{,}}
\end{equation}
where $\tilde \theta \left( t \right) = \hat \theta \left( t \right) - \theta$ stands for a parametric error, and $T > 0$ is a tuning gain of the estimation algorithm (sliding window width under the condition of the regressor persistent excitation).

Therefore, it is required to ensure asymptotic convergence of parametric error to an arbitrarily small neighborhood of zero, which is adjustable by the tuning parameter of the estimation law. This problem will be solved under the following assumptions.

\begin{assumption}$\hat \varphi  \in {\rm{ST}} \cap {\rm{PE}}$ and {$\hat{\varphi}\in L_\infty, w\in L_\infty$}.
\end{assumption}
\begin{assumption} There exists a scalar ${\rm{0}} < m \le n$ and known matrices ${{\cal L}_1} \in {\mathbb{R}^{n \times m}}{\rm{,\;}}{{\cal L}_2} \in {\mathbb{R}^{n \times \left( {n - m} \right)}}$ such that: 
\begin{equation}\label{6}
{{\cal L}_1}{\cal L}_1^{\top} + {{\cal L}_2}{\cal L}_2^{\top} = {I_n}{\rm{,}}
\end{equation}
and additionally\footnote{ {Note that independence and dependence notions are mutually exclusive but not exhaustive, {\emph{i.e}}., there exist signals that match neither definition 3 nor 4. Hence, the existence of such a decomposition is an {\emph{a priori}} strict structural-informational assumption on the LRE \eqref{3}.}}: 
\begin{enumerate}
\item[a)] for ${\hat \varphi ^{{\rm{ind}}}}\left( t \right){\rm{:}} = {\cal L}_1^{\top}\hat \varphi \left( t \right)$, it holds that ${\hat \varphi ^{{\rm{ind}}}} \bot w$,
\item [b)] for ${\hat \varphi ^{{\rm{dep}}}}\left( t \right){\rm{:}} = {\cal L}_2^{\top}\hat \varphi \left( t \right)$,  it holds that ${\hat \varphi ^{{\rm{dep}}}}\parallel w$.
\end{enumerate}
\end{assumption}

Assumption 2 is a condition of independence of $m$ elements of the regressor $\hat \varphi \left( t \right)$ from the additive perturbation $w\left( t \right)$. To fulfill this assumption, it is required that the noise ${\eta _\varphi }\left( t \right)$ does not affect all elements of the regressor $\varphi \left( t \right)$, and there exist elements of the regressor that are independent from ${\eta _f}\left( t \right)$ and ${\eta _z}\left( t \right)$. The matrices ${{\cal L}_1}{\rm{,\;}}{{\cal L}_2}$ are chosen for each particular case on the basis of \emph{a priori} knowledge via \emph{direct inspection} of each specific parametrization \eqref{1}. Let us illustrate the essence of the notation and assumptions introduced in the problem statement by the following example.

\textbf{Example 1.} Consider a first-order linear dynamical system
\begin{equation}\label{7}
\begin{array}{l}
\dot x\left( t \right) = ax\left( t \right) + bu\left( t \right) + f\left( t \right){\rm{,\;}}x\left( {{t_0}} \right) = 0,\\
y\left( t \right) = x\left( t \right) + {\eta _x}\left( t \right){\rm{,}}
\end{array}
\end{equation}
where $x\left( t \right)$ is a system state, $f\left( t \right)$ denotes an external perturbation, ${\eta _x}\left( t \right)$ stands for a measurement noise, and \linebreak $\theta  = { {\begin{bmatrix}
a&b
\end{bmatrix}}^{\top}}$ is a vector of unknown parameters.

A parametrization of the system \eqref{7} in the form of the linear regression equation \eqref{1} is written as:
\begin{equation*}
\begin{array}{l}
z\left( t \right){\rm{:}} = {\textstyle{s \over {ks + 1}}}\left[ x \right]\left( t \right)-{\textstyle{1 \over {ks + 1}}}\left[ f \right]\left( t \right){\rm{,\;}}\\\varphi \left( t \right){\rm{:}} = {{\begin{bmatrix}
{{\textstyle{1 \over {ks + 1}}}\left[ x \right]\left( t \right)}&{{\textstyle{1 \over {ks + 1}}}\left[ u \right]\left( t \right)}
\end{bmatrix}}^{\top}}{\rm{, }}
\end{array}
\end{equation*}	
where $k > 0$ and ${{s}} {\rm{: = }}{\textstyle{d \over {dt}}}$.

As external perturbation and noise are not available for measurement, we have the following observation equations
\begin{equation*}
\begin{array}{l}
\hat z\left( t \right){\rm{:}} = {\textstyle{s \over {ks + 1}}}\left[ y \right]\left( t \right) = z\left( t \right){\rm{ + }}\underbrace {{\textstyle{s \over {ks + 1}}}\left[ {{\eta _x}} \right]\left( t \right)}_{{\eta _z}\left( t \right)} + \underbrace {{\textstyle{1 \over {ks + 1}}}\left[ f \right]\left( t \right)}_{{\eta _f}\left( t \right)}{\rm{, }}\\
\hat \varphi \left( t \right){\rm{:}} = {{\begin{bmatrix}
{{\textstyle{1 \over {ks + 1}}}\left[ y \right]\left( t \right)}&{{\textstyle{1 \over {ks + 1}}}\left[ u \right]\left( t \right)}
\end{bmatrix}}^{\top}}{\rm{ =  }}\\\hfill=\varphi \left( t \right) + \underbrace {{{ {\begin{bmatrix}
{{\textstyle{1 \over {ks + 1}}}\left[ {{\eta _x}} \right]\left( t \right)}&0
\end{bmatrix}}}^{\top}}}_{{\eta _\varphi }\left( t \right)}.
\end{array}
\end{equation*}

For simplicity and without loss of generality, let us assume that the noise ${\eta _x}\left( t \right)$ and external disturbance $f\left( t \right)$ are represented as multiharmonic series with a finite number of terms. Then, if $u\left( t \right)$ is an external (feedforward) signal described by a multiharmonic series with a finite number of frequencies, then $u\left( t \right)$ contains no information about external disturbance $f\left( t \right)$ and noise ${\eta _x}\left( t \right)$, which means that $u\left( t \right)$ and $w\left( t \right)$ spectra have no common frequencies, and, therefore, by Proposition 3, Assumption 2 is met for ${\cal L}_1^{\top} = {\begin{bmatrix}0&1 \end{bmatrix}}$. On the other hand, the output signal $y\left( t \right)$ depends explicitly from measurement noise ${\eta _x}\left( t \right)$, and implicitly – from an external disturbance $f\left( t \right)$, and, therefore, by Proposition 3, Assumption 2 is not met for ${\cal L}_1^{\top} = {\begin{bmatrix}1&0\end{bmatrix}}$. If we assume that the signal $u\left( t \right)$ is an output of a feedback controller $u\left( t \right) =  - Ky\left( t \right){\rm{,\;}}K > 0$, then the second element of the measurable regressor $\hat \varphi \left( t \right)$ depends explicitly from measurement noise and implicitly – from the external disturbance, which means that Assumption 2 is not met, and all elements of the regressor depend from the additive disturbance. Therefore, in general case, Assumption 2 holds only if the regressor contains $m$ external signals, \emph{i.e.}, external inputs of the system/parameterization, which have no relation to system/parameterization internal states or their measurements.$\hfill\blacksquare$

\section{Main Result}

In this section, new estimation algorithms are proposed that allow one to achieve the stated goal \eqref{5} under various assumptions. Section A is to consider the case $m=n$, \emph{i.e.}, when all the elements of the regressor are independent from the additive perturbation. In Section B, the obtained results are extended to the case when ${{0}} < {{m}}  < n$ and condition $\theta  \in {D_\theta }{\rm{:}} = \left\{ {\theta  \in {\mathbb{R}^n}{\rm{:\;}}{{\cal H}^{\top}}\theta  = {0_{n - m}}} \right\}$ holds for some ${\cal H} \in {\mathbb{R}^{n \times \left( {n - m} \right)}}$. Section C also discusses the case ${{0}} < {{m}}  < n$, but the condition $\theta  \in {D_\theta }$ is replaced by less restrictive assumptions. 
\subsection{Fully independence case}

The fact that the condition $\hat \varphi  \bot w$ is met motivates us to introduce the following regressor extension scheme. Equation \eqref{3} is multiplied by $\hat \varphi \left( t \right)$, and the obtained result is integrated over a sliding window $\left[ {t - T{\rm{,\;}}t} \right]$ for all $t \ge {t_0}$ to have: 
\begin{equation*}
\begin{array}{l}
    {\textstyle{1 \over T}}\int\limits_{{\rm{max}}\left\{ {{t_0}{\rm{, }}t - T} \right\}}^t {\hat \varphi \left( \tau  \right)\hat z\left( \tau  \right)d\tau }  =\\ \hfill {\textstyle{1 \over T}}\!\!\!\int\limits_{{\rm{max}}\left\{ {{t_0}{\rm{, }}t - T} \right\}}^t \!\!\!{\hat \varphi \left( \tau  \right){{\hat \varphi }^{\top}}\left( \tau  \right)d\tau } \theta  + {\textstyle{1 \over T}}\!\!\!\int\limits_{{\rm{max}}\left\{ {{t_0}{\rm{, }}t - T} \right\}}^t \!\!\!{\hat \varphi \left( \tau  \right)w\left( \tau  \right)d\tau } {\rm{,}}
    \end{array}
\end{equation*}
or in a more compact form\footnote{As we integrate over a finite window, then boundedness of $Y\left( t \right){\rm{,\;}}\Phi \left( t \right){\rm{,\;}}W\left( t \right)$ follows from the boundedness of $\hat z\left( t \right){\rm{,\;}}\hat \varphi \left( t \right){\rm{,\;}}w\left( t \right)$.}:
\begin{equation}\label{8}
Y\left( t \right) = \Phi \left( t \right)\theta  + W\left( t \right){\rm{,}}
\end{equation}
where, for all $t \ge {t_0}$, $Y\left( t \right) \in {\mathbb{R}^n}$ and $\Phi \left( t \right) \in {\mathbb{R}^{n \times n}}$ are the solutions of
\begin{equation}\label{9}
\begin{array}{l}
\dot Y\left( t \right) = {\textstyle{1 \over T}}\left[ {\hat \varphi \left( t \right)\hat z\left( t \right) - \hat \varphi \left( {t - T} \right)\hat z\left( {t - T} \right)} \right]{\rm{,\;}}\\
\dot \Phi \left( t \right) = {\textstyle{1 \over T}}\left[ {\hat \varphi \left( t \right){{\hat \varphi }^{\top}}\left( t \right) - \hat \varphi \left( {t - T} \right){{\hat \varphi }^{\top}}\left( {t - T} \right)} \right]{\rm{,\;}}
\end{array}
\end{equation}
and $W\left( t \right)$ is a solution of:
\begin{equation*}
\dot W\left( t \right) = {\textstyle{1 \over T}}\left[ {\hat \varphi \left( t \right)w\left( t \right) - \hat \varphi \left( {t - T} \right)w\left( {t - T} \right)} \right]{\rm{,\;}}
\end{equation*}
with $Y\left( {{t_0}} \right) = {0_n}{\rm{,}}$ $\Phi \left( {{t_0}} \right) = {0_{n \times n}}{\rm{,}}$ $W\left( {{t_0}} \right) = {0_n}.$

Owing to Proposition 2, Assumptions 1 and 2 mean that the following limits exist:
\begin{equation*}{
\begin{array}{c}
\mathop {{\rm{lim}}}\limits_{T \to \infty } \mathop {{\rm{sup}}}\limits_{t \geqslant {t_0} + T} \left\| {\Phi \left( t \right) - \Lambda } \right\| = 0,\\
\mathop {{\rm{lim}}}\limits_{T \to \infty } \mathop {{\rm{sup}}}\limits_{t \geqslant {t_0} + T} \left\| {W\left( t \right)} \right\| = 0.
\end{array}}
\end{equation*}

Therefore, if Assumptions 1 and 2 are met simultaneously, then it holds that 
\begin{equation*}{
\mathop {{\rm{lim}}}\limits_{T \to \infty } \mathop {{\rm{sup}}}\limits_{t \geqslant {t_0} + T} \left\| {{\Phi ^{ - 1}}\left( t \right)Y\left( t \right) - \theta } \right\| = 0.}
\end{equation*}
and, using measurable signals \eqref{9}, an estimation law can be designed that ensures that the stated goal \eqref{5} is met. 

{In this study, it is proposed to apply the following DREM\cite{11} based estimation law:}
\begin{equation}\label{10}
\dot {\hat \theta} \left( t \right) =  - \gamma {\rm{adj}}\left\{ {\Phi \left( t \right)} \right\}\left( {\Phi \left( t \right)\hat \theta \left( t \right) - Y\left( t \right)} \right),
\end{equation}
where $\hat \theta \left( {{t_0}} \right) = {\hat \theta _0}$, $\gamma > 0$.

Formally, the conditions of such estimation law convergence are formulated and proved in the following theorem (proof is postponed to Supplementary material \cite{27}).
\begin{theorem}
    If Assumptions 1 and 2 are met, then in case $m=n$ the law \eqref{10} ensures that the stated goal \eqref{5} is achieved.
\end{theorem}

Therefore, if Assumptions 1, 2 are met and $m=n$ regressor elements are independent from the perturbation, then application of the sliding window filter \eqref{9} for the dynamic regressor extension ensures convergence of the parametric error to an arbitrarily small neighborhood of zero, which depends on the width of the sliding window in use. 
If $m < n$ regressor elements are independent from the additive perturbation, then the perturbation is represented as:
\begin{equation}\label{15}
W\left( t \right) = {W_1}\left( t \right) + {W_2}\left( t \right){\rm{,}}
\end{equation}
where ${W_1}\left( t \right){\rm{:}} = {{\cal L}_1}{\cal L}_1^{\top}W\left( t \right){\rm{,\;}}{W_2}\left( t \right){\rm{:}} = {{\cal L}_2}{\cal L}_2^{\top}W\left( t \right).$

In this case, only the component ${W_1}\left( t \right)$ can be reduced by increase of the filter \eqref{9} parameter $T>0$, and it is impossible to reduce the component ${W_2}\left( t \right)$, and therefore, in this case, the estimation law \eqref{10} does not ensure achievement of the stated goal. The properties of the estimation law \eqref{10} for this case are strictly shown in the following theorem.
\begin{theorem}
    If Assumptions 1, 2 are met, then, in case $m < n$, the estimation law \eqref{10} ensures for any $\varepsilon>0$ that there exists $\delta > 0$ such that, if $T > \delta$, then it holds for all $t \ge {t_0} + T$ that:
\begin{equation*}
\begin{array}{l}
\left| {{{\tilde \theta }_i}\left( t \right)} \right| \le {e^{ - \mu \left( {t - {t_0} - T} \right)}}\left| {{{\tilde \theta }_i}\left( {{t_0} + T} \right)} \right| + \\\hfill\quad\quad\quad\quad\quad\quad\quad\quad\quad\quad\quad\quad + \sqrt {{\textstyle{\rho  \over \mu }}C_W^{\top}{{\cal L}_2^{\top}{\cal L}_2}{C_W}}  + \sqrt {2{\textstyle{\rho  \over \mu }}} \varepsilon {\rm ,}
\end{array}
\end{equation*}
where $\varepsilon  > 0$ stands for an arbitrarily small scalar, and the scalars $\rho  > 0,{\rm{\;}}\mu  > 0$ are defined in the proof. 
\end{theorem}
\begin{proof}Proof is given in Supplementary material \cite{27}.\end{proof}

    Considering the upper bound from Theorem 2, the ratio ${\textstyle{\rho  \over \mu }}$ is a function of the matrix $\Lambda  \in {\mathbb{R}^{n \times n}}$, and therefore, this ratio cannot be made arbitrarily small for a fixed value of $\Lambda$. Therefore, the bound $\sqrt {{\textstyle{\rho  \over \mu }}C_W^{\top}{{\cal L}_2^{\top}{\cal L}_2}{C_W}}$ is completely determined by the ratio between the autocovariance $\Lambda  \in {\mathbb{R}^{n \times n}}$ and cross-correlation ${C_W} \in {\mathbb{R}^n}$ (if the minimum eigenvalue of the matrix $\Lambda$ is ``large'' and $C_W^{\top}{C_W}$ is ``small'', then the bound $\sqrt {{\textstyle{\rho  \over \mu }}C_W^{\top}{{\cal L}_2^{\top}{\cal L}_2}{C_W}}$ is small). Thus, Theorem 2 shows that the estimation law \eqref{10} allows one to eliminate the component ${W_1}\left( t \right)$ up to an arbitrarily small error $\varepsilon$, and, at the same time, such a law is robust in the UUB sense with respect to ${W_2}\left( t \right)$.

\subsection{Partial independence case (Scheme 1)}

The analysis from previous subsection shows that, when only $m < n$ regressor elements are independent from the disturbance, the component ${W_2}\left( t \right)$ prevents the achievement of the stated goal \eqref{5}. In order to annihilate such distortion, equation \eqref{8} is multiplied by ${\rm{adj}}\left\{ {\Phi \left( t \right)} \right\}$ to have
\begin{equation}\label{20}
\begin{array}{l}
{\cal Y}\left( t \right) = \Delta \left( t \right)\theta  +\hfill\\\hfill\quad\quad\quad+ {\rm{adj}}\left\{ {\Phi \left( t \right)} \right\}\left[ {{{\cal L}_1}{\cal L}_1^{\top}W\left( t \right) + {{\cal L}_2}{\cal L}_2^{\top}W\left( t \right)} \right]{\rm{,}}
\end{array}
\end{equation}
where ${\cal Y}\left( t \right){\rm{:}} = {\rm{adj}}\left\{ {\Phi \left( t \right)} \right\}Y\left( t \right){\rm{,\;}}\Delta \left( t \right){\rm{:}} = {\rm{det}}\left\{ {\Phi \left( t \right)} \right\}.$

The following assumption is adopted for further analysis.

\begin{assumption}
    $\theta  \in {D_\theta }{\rm{:}} = \left\{ {\theta  \in {\mathbb{R}^n}{\rm{:\;}}{{\cal H}^{\top}}\theta  = {0_{n - m}}} \right\}$, where the matrix ${\cal H} \in {\mathbb{R}^{n \times \left( {n - m} \right)}}$  is known and satisfies the condition ${\rm{rank}}\left\{ {{{\cal H}^{\top}}{\rm{adj}}\left\{ \Lambda  \right\}{{\cal L}_2}} \right\} = n - m$ for {$\Lambda \in \mathbb{R}^{n \times n}$ such that $\mathop {{\rm{lim}}}\limits_{T \to \infty } \mathop {{\rm{sup}}}\limits_{t \geqslant {t_0} + T} \left\| {\Phi \left( t \right) - \Lambda } \right\| = 0.$}
\end{assumption}

If Assumption 3 is met, then multiplication of \eqref{20}, firstly, by ${{\cal H}^{\top}}$ and then by ${{\cal L}_2}{\rm{adj}}\left\{ {{{\cal H}^{\top}}{\rm{adj}}\left\{ {\Phi \left( t \right)} \right\}{{\cal L}_2}} \right\}$ allows one to obtain:
\begin{equation}\label{21}
{\cal N}\left( t \right) = {\cal M}{_1}\left( t \right){W_1}\left( t \right) + {{\cal M}_2}\left( t \right){W_2}\left( t \right){\rm{,}}
\end{equation}
where 
\begin{equation*}
\begin{array}{l}
{\cal N}\left( t \right){\rm{:}} = {{\cal L}_2}{\rm{adj}}\left\{ {{{\cal H}^{\top}}{\rm{adj}}\left\{ {\Phi \left( t \right)} \right\}{{\cal L}_2}} \right\}{{\cal H}^{\top}}{\cal Y}\left( t \right){\rm{,}}\\
{{\cal M}_1}\left( t \right){\rm{:}} = {{\cal L}_2}{\rm{adj}}\left\{ {{{\cal H}^{\top}}{\rm{adj}}\left\{ {\Phi \left( t \right)} \right\}{{\cal L}_2}} \right\}{{\cal H}^{\top}}{\rm{adj}}\left\{ {\Phi \left( t \right)} \right\}{\rm{,}}\\
{{\cal M}_2}\left( t \right){\rm{:}} = {\rm{det}}\left\{ { {{{\cal H}^{\top}}{\rm{adj}}\left\{ {\Phi \left( t \right)} \right\}{{\cal L}_2}}} \right\}.
\end{array}
\end{equation*}

Then, having multiplied \eqref{8} by ${{\cal M}_2}\left( t \right)$ and subtracting \eqref{21} from the obtained result, we have:
\begin{equation}\label{22}
\Upsilon \left( t \right) = {{\cal M}_2}\left( t \right)\Phi \left( t \right)\theta  + {\cal{M}}\left(t\right){W_1}\left( t \right){\rm{,}}
\end{equation}
where
\begin{equation*}
\begin{array}{c}
\Upsilon \left( t \right) = {{\cal M}_2}\left( t \right)Y\left( t \right) - {\cal N}\left( t \right),\\
{\cal{M}}\left(t\right)={{{\cal M}_2}\left(t\right){I_n} - {{\cal M}_1}\left(t\right)}.
\end{array}
\end{equation*}

Therefore, owing to Assumption 3, a new regression equation \eqref{22}, which does not contain a term ${W_2}\left( t \right)$, is obtained from \eqref{8}. At the same time, the component ${W_1}\left( t \right)$ is multiplied by a bounded term (when Assumption 1 is met), and still, when Assumption 2 is satisfied, it converges to zero asymptotically w.r.t. parameter $T > 0$ and uniformly w.r.t. time. Therefore, if Assumptions 1-3 are met simultaneously, on the basis of measurable signals $\Upsilon \left( t \right){\rm{,\;}}{{\cal M}_2}\left( t \right){\rm{,\;}}\Phi \left( t \right)$, an estimation law can be designed, which guarantees that the stated goal \eqref{5} is achieved. We apply the following estimation law:
\begin{equation}\label{23}
\dot {\hat \theta} \left( t \right)\!=\!- \gamma {{\cal M}_2}\left( t \right){\rm{adj}}\left\{ {\Phi \left( t \right)} \right\}\left( {{{\cal M}_2}\left( t \right)\Phi \left( t \right)\hat \theta\!-\!\Upsilon \left( t \right)} \right){\rm{,}}
\end{equation}
where $\hat \theta \left( {{t_0}} \right) = {\hat \theta _0}$ and $\gamma>0$.

The convergence conditions for such an estimation law are strictly shown and proved in the following theorem.
\begin{theorem}
    If Assumptions 1-3 are met, then the estimation law \eqref{23} ensures achievement of the stated goal \eqref{5}.
\end{theorem}
\begin{proof}Proof is presented in Supplementary material \cite{27}.\end{proof}

In comparison with \eqref{10}, the estimation law \eqref{23} ensures that the goal \eqref{5} is achieved in case $m < n$ rather than $m=n$. However, in order to meet \eqref{5}, the law \eqref{23} requires adoption of additional Assumption 3, which is difficult to satisfy and verify. If the parameters $\theta$ are obtained via overparameterization, \emph{i.e.} $\theta {\rm{: = }}{\cal G}\left( \eta  \right){\rm{,\;}}\eta  \in {\mathbb{R}^{{n_\eta}}}$, ${n_\eta } \le n$, then, using the structure of the mapping ${\cal G}{\rm{:\;}}{\mathbb{R}^{{n_\eta }}} \mapsto {\mathbb{R}^n}$, it is sometimes possible to satisfy the requirements of Assumption 3.

\textbf{Example 2.} If $a =  - b$ in the system \eqref{7}, then, for some $\Lambda = {\Lambda ^{\top}} > 0$, Assumption 3 is met for ${{\cal H}^{\top}} = {\begin{bmatrix}
1&1
\end{bmatrix}}$, ${\cal L}_2^{\top} = {\begin{bmatrix}
1&0
\end{bmatrix}}.$ $\hfill\blacksquare$
\subsection{Partial independence case (Scheme 2)}

It is obvious that, in general, Assumption 3 is impossible to be met without {\emph{a priori}} knowledge of the parameters to select apropriate annihilator ${\cal{H}}^{\top}$. This assumption is proposed to be relaxed via transformation of the regression equation \eqref{1} into a new one with respect to a new extended vector of parameters $\Theta  = {\cal D}\theta$ such that: a) for $\Theta$, Assumption 3 is met owing to definition of the matrix ${\cal D} \in {\mathbb{R}^{2n \times n}}$, b) ${\rm{rank}}\left\{ {\cal D} \right\} = n$, {\emph{i.e.}}, the estimation of $\Theta  \in {\mathbb{R}^{2n}}$ is equivalent to the estimation of the original parameters $\theta  \in {\mathbb{R}^n}$.

To achieve this, the filter $W\left( s \right)\left[ . \right]$ ($W\left( s \right)$ is a stable {strictly-proper} minimum-phase transfer function) is applied to \eqref{3}:
\begin{equation}\label{28}
{z_f}\left( t \right) = {\varphi _f}\left( t \right)\theta  + {w_f}\left( t \right){\rm{,}}
\end{equation}
where
\begin{equation*}
\begin{array}{c}
{z_f}\left( t \right){\rm{:}} = W\left( s \right)\left[ \hat{z} \right]\left( t \right){\rm{,\;}}{w_f}\left( t \right){\rm{:}} = W\left( s \right)\left[ w \right]\left( t \right){\rm{,}}\\
{\varphi _f}\left( t \right){\rm{:}} = {{\begin{bmatrix}
{W\left( s \right)\left[ {{{\hat \varphi }_1}} \right]\left( t \right)}& \cdots &{W\left( s \right)\left[ {{{\hat \varphi }_n}} \right]\left( t \right)}
\end{bmatrix}}^{\top}}{\rm{. }}
\end{array}
\end{equation*} 

Having subtracted \eqref{28} from \eqref{3}, it is obtained that:
\begin{equation}\label{29}
\tilde z\left( t \right) = \phi \left( t \right)\Theta  + \tilde w\left( t \right){\rm{,}}
\end{equation}
where
\begin{equation}
\begin{array}{c}
\tilde z\left( t \right){\rm{:}} = \hat{z}\left( t \right) - {z_f}\left( t \right){\rm{,\;}}\tilde w\left( t \right){\rm{:}} = w\left( t \right) - {w_f}\left( t \right){\rm{,}}\\
\phi \left( t \right){\rm{:}} = {{\begin{bmatrix}
{\hat \varphi \left( t \right)}&{{\varphi _f}\left( t \right)}
\end{bmatrix}}^{\top}}
\end{array}
\end{equation} 
and ${\cal D} = {{\begin{bmatrix}{{I_n}}&{ - {I_n}}
\end{bmatrix}}^{\top}}$.

The main idea of transformation \eqref{28}, \eqref{29} is to meet artificially the main requirement of Assumption 3 for an extended vector of unknown parameters $\Theta$. By construction of \eqref{29}, the following proposition holds (proof is given in Supplementary material \cite{27}).
\begin{proposition}
    If $2m \ge n$, then, for ${\cal D} = {{\begin{bmatrix}
{{I_n}}&{ - {I_n}}
\end{bmatrix}}^{\top}}$, there exists ${{\cal H}_e} \in {\mathbb{R}^{2n \times 2\left( {n - m} \right)}}$ such that:
\begin{enumerate}
\item[1)] ${\cal H}_e^{\top}{\cal D} = {0_{2\left( {n - m} \right) \times n}}{\rm{,}}$
\item[2)] ${\rm{rank}}\left\{ {{\cal H}_e^{\top}} \right\} = 2\left( {n - m} \right).$
\end{enumerate}
\end{proposition}

Thus, according to Proposition 5, a matrix \linebreak ${\cal H}_e^{\top} \in {\mathbb{R}^{2\left( {n - m} \right) \times 2n}}$ can be defined, which is an annihilator for a given matrix ${\cal D} \in {\mathbb{R}^{2n \times n}}$. However, the resulting regression equation \eqref{29} contains a new regressor $\phi \left( t \right)$ and a new perturbation $\tilde{w}\left( t \right)$, and, in general case, it does not follow from $\hat \varphi  \in {\rm{PE}}$ that the condition $\phi  \in {\rm{PE}}$ is met (a counter example is the case when ${\hat \varphi _i}\left( t \right) = const$ for some $i \in \left\{ {1, \ldots n} \right\}$. Therefore, for further design of the estimation law on the basis of the obtained equation \eqref{29}, it is proposed to adopt the following assumption.
\begin{assumption}
$\hat \varphi  \in {\rm{ST}} \cap {\rm{PE}} \Rightarrow \phi  \in {\rm{ST}} \cap {\rm{PE}}{\rm{.}}$
\end{assumption}

{On the other hand, the propagation of independence and dependence properties through transformations \eqref{28}, \eqref{29} is strictly investigated in the following proposition, which is obtained on the basis of Assumption 2 and Proposition 4.}

\begin{proposition}
    If Assumptions 1 and 2 are met, then: 
\begin{enumerate}
\item[a)] ${\phi ^{{\rm{ind}}}} \bot \tilde w$ holds for ${\phi ^{{\rm{ind}}}}\left( t \right){\rm{:}} = {\cal L}_{1e}^{\top}\phi \left( t \right)$,
\item[b)] ${\phi ^{{\rm{dep}}}}\parallel \tilde w$ holds for ${\phi ^{{\rm{dep}}}}\left( t \right){\rm{:}} = {\cal L}_{2e}^{\top}\phi \left( t \right)$,
\end{enumerate}
where 
\begin{equation*}{{\cal L}_{1e}} =  {\begin{bmatrix}
{{{\cal L}_1}}&{{0_{n \times m}}}\\
{{0_{n \times m}}}&{{{\cal L}_1}}
\end{bmatrix}}{\rm{,\;}}{{\cal L}_{2e}} = {\begin{bmatrix}
{{{\cal L}_2}}&{{0_{n \times \left( {n - m} \right)}}}\\
{{0_{n \times \left( {n - m} \right)}}}&{{{\cal L}_2}}
\end{bmatrix}}.\end{equation*}
\end{proposition}

\begin{proof}Proof is presented in Supplementary material \cite{27}.\end{proof}

Therefore, being motivated by the results from Subsection B, adoption of Assumptions 1, 2, and 4 allows one to introduce the following estimation law for the original parameters, which is based on equation \eqref{29} with an extended vector of parameters:
\begin{equation}\label{30}
\begin{array}{l}
{\dot {\hat \theta}_i}\left( t \right)\!=\!\Gamma_i\left(t\right)\left({{{\cal M}_2}\left( t \right)\Phi \left( t \right)\mathcal{D}\hat \theta\left(t\right)\!-\!\Upsilon \left( t \right)} \right){\rm{,}}\\
\Gamma_i\left(t\right)=- \gamma {{\cal M}_2}\left( t \right)\!{{E}}_i^{\top}{\rm{adj}}\left\{ {\Phi \left( t \right)}\right\}{\rm{,}}
\end{array}
\end{equation}
where
\begin{equation}\label{31}
\begin{array}{c}
\Upsilon \left( t \right) = {{\cal M}_2}\left( t \right)Y\left( t \right) - {\cal N}\left( t \right){\rm{,}}\\
{\cal N}\left( t \right){\rm{:}} = {{\cal L}_{2e}}{\rm{adj}}\left\{ {{\cal H}_e^{\top}{\rm{adj}}\left\{ {\Phi \left( t \right)} \right\}{{\cal L}_{2e}}} \right\}{{\cal H}_e^{\top}}{\cal Y}\left( t \right){\rm{,}}\\
\!{{\cal M}_1}\left( t \right){\rm{:}}\!=\!{{\cal L}_{2e}}{\rm{adj}}\left\{ {{\cal H}_e^{\top}{\rm{adj}}\left\{ {\Phi} \right\}{{\cal L}_{2e}}} \right\}{\cal H}_e^{\top}{\rm{adj}}\left\{ {\Phi \left( t \right)} \right\}{\rm{,}}\!\\
{{\cal M}_2}\left( t \right){\rm{:}} = {\rm{det}}\left\{ { {{\cal H}_e^{\top}{\rm{adj}}\left\{ {\Phi \left( t \right)} \right\}{{\cal L}_{2e}}}} \right\}{\rm{,}}\\
{\cal Y}\left( t \right){\rm{:}} = {\rm{adj}}\left\{ {\Phi \left( t \right)} \right\}Y\left( t \right){\rm{,\;}}\Delta \left( t \right){\rm{:}} = {\rm{det}}\left\{ {\Phi \left( t \right)} \right\}{\rm{,}}
\end{array}
\end{equation}
\begin{equation}\label{32}
\begin{array}{l}
\dot Y\left( t \right) = {\textstyle{1 \over T}}\left[ {\phi \left( t \right)\tilde z\left( t \right) - \phi \left( {t - T} \right)\tilde z\left( {t - T} \right)} \right]{\rm{,\;}}\\
\dot \Phi \left( t \right) = {\textstyle{1 \over T}}\left[ {\phi \left( t \right){\phi ^{\top}}\left( t \right) - \phi \left( {t - T} \right){\phi ^{\top}}\left( {t - T} \right)} \right]{\rm{,\;}}
\end{array}
\end{equation}
and $\hat \theta \left( {{t_0}} \right) = {\hat \theta _0}$, $Y\left( {{t_0}} \right) = {0_{2n}}{\rm{,}}$ $\Phi \left( {{t_0}} \right) = {0_{2n \times 2n}}{\rm{,}}$ $\gamma > 0$, $E_i^{\top} = e_i^{\top}{\begin{bmatrix}
{{I_{n \times n}}}&{{0_{n \times n}}}
\end{bmatrix}}.$

The conditions, under which such an estimation law ensures achievement of the stated goal, are given in the following theorem.
\begin{theorem}
    Let $2m \ge n$, Assumptions 1, 2 and 4 are met and, additionally, for {$\Lambda_e\in\mathbb{R}^{2n\times 2n}$ such that}
\begin{equation*} {
\begin{array}{l}
\mathop {{\rm{lim}}}\limits_{T \to \infty } \mathop {{\rm{sup}}}\limits_{t \geqslant {t_0}} \left\| {\frac{1}{T}\int\limits_t^{t + T} {\phi \left( \tau  \right){\phi ^ \top }\left( \tau  \right)} d\tau  - {\Lambda _e}} \right\|=0
\end{array}}
\end{equation*}
and defined ${{\cal H}_e} \in {\mathbb{R}^{2n \times 2\left( {n - m} \right)}}{\rm{,\;}}{{\cal L}_{2e}} \in {\mathbb{R}^{2n \times 2\left( {n - m} \right)}}$ (see Propositions 5 and 6) it holds that 

\begin{equation*}{\rm{rank}}\left\{ {{\cal H}_e^{\top}{\rm{adj}}\left\{ {{\Lambda _e}} \right\}{{\cal L}_{2e}}} \right\} = 2\left( {n - m} \right).\end{equation*}

Then the law \eqref{30} ensures achievement of the stated goal \eqref{5}.
\end{theorem}
\begin{proof}Proof is postponed to Supplementary material \cite{27}.\end{proof}

In comparison with \eqref{10}, the estimation law \eqref{30} ensures that the goal \eqref{5} is achieved in case ${\textstyle{n \over 2}} \le m < n$ rather than $m=n$, while in comparison with \eqref{23}, instead of a restrictive condition ${{\cal H}^{\top}}\theta  = {0_{n - m}}$, it requires to meet Assumption 4 about propagation of the regressor persistent excitation. In authors’ opinion, Assumption 4 is less restrictive than Assumption 3 and very often fulfilled in the practical scenarios.

\textbf{Example 3.} Consider a linear closed-loop dynamical system 
\begin{equation*}
\begin{array}{l}
y\left( t \right) = \frac{{Z\left( {\theta {\rm{,\;}}s} \right)}}{{R\left( {\theta {\rm{,\; }}s} \right)}}\left[ {u} \right]\left( t \right)+f\left( t \right){\rm{,}}\\
\hat{y}\left(t\right)=y\left(t\right)+\eta_{y}\left(t\right){\rm{,}}
\end{array}
\end{equation*}
where 
\begin{equation*}
\begin{array}{l}
Z\left( {\theta {\rm{,\;}}s} \right) = {b_{q}}{s^{q}} + {b_{q - 1}}{s^{q - 1}} +  \ldots  + {b_0}{\rm{,}}\\
R\left( {\theta {\rm{,\;}}s} \right) = {s^p} + {a_{p - 1}}{s^{p - 1}} +  \ldots  + {a_0}{\rm{,}}
\end{array}
\end{equation*}
and $y\left(t\right)$ is a system output, $\hat{y}\left(t\right)$ is an output available for measurement, $u\left(t\right)$ is an input (reference value), $f\left(t\right)$ is an external perturbation, $\eta_y\left(t\right)$ is a measurement noise. Without loss of generality, we assume that the degree of $Z\left(\theta,\;s\right)$ is {\emph{unknown}} and consider the worst case {$q = p-1$}. If $Z\left( {\theta {\rm{,\;}}s} \right)$ is of known degree $q < p - 1$, then the coefficients ${b_i}{\rm{,\;}}i =\linebreak = p - 1, \ldots {{,}}{q +1}$ are simply equaled to zero.

Then the signals $\hat{z}\left( t \right)$ and $\hat{\varphi}\left( t \right)$ are obtained as follows \cite{1,2}:
\begin{equation*}
\begin{array}{l}
\hat{z}\left( t \right) {\rm{:}} = {\frac{{{s^{{p}}}}}{{\delta \left( s \right)}}}\left[\hat{y}\right]\left( t \right) = {\frac{{{s^{{p}}}}}{{\delta \left( s \right)}}}\left[y\right]\left( t \right)+\underbrace {{\textstyle{{{s^p}} \over {\delta \left( s \right)}}}\left[ {{\eta _y}} \right]\left( t \right)}_{{\eta _z}\left( t \right)}\\\hfill\quad= \underbrace { {\textstyle{{{s^{{p}}}} \over {\delta \left( s \right)}}}\left[y\right]\left( t \right) - {\textstyle{\textstyle{{{R\left( {\theta {\rm{,\;}}s} \right)}}}\over {\delta \left( s \right)}}}\left[ f \right]\left( t \right)}_{{z}\left( t \right)}+\eta _z\left( t \right)+\underbrace { {\textstyle{\textstyle{{{R\left( {\theta {\rm{,\;}}s} \right)}}}\over {\delta \left( s \right)}}}\left[ f \right]\left( t \right)}_{{\eta _f}\left( t \right)}{\rm{,}}\\
\hat{\varphi} \left( t \right) {\rm{:}} = {{\begin{bmatrix}
{ - {\textstyle{{\lambda _{{p} - 1}^{\top}\left( s \right)} \over {\delta \left( s \right)}}}\left[\hat{y}\right]\left( t \right)}&{{\textstyle{{\lambda _{{p} - 1}^{\top}\left( s \right)} \over {\delta \left( s \right)}}}\left[u\right]\left( t \right)}
\end{bmatrix}}^{\top}}=\\
\hfill=\underbrace{{{\begin{bmatrix}
{ - {\textstyle{{\lambda _{{p} - 1}^{\top}\left( s \right)} \over {\delta \left( s \right)}}}\left[y\right]\left( t \right)}&{{\textstyle{{\lambda _{{p} - 1}^{\top}\left( s \right)} \over {\delta \left( s \right)}}}\left[u\right]\left( t \right)}
\end{bmatrix}}^{\top}}}_{\varphi\left(t\right)}+\\\hfill+
\underbrace {{{{\begin{bmatrix}
{ - {\textstyle{{\lambda _{p - 1}^ \top \left( s \right)} \over {\delta \left( s \right)}}}\left[ {{\eta _y}} \right]\left( t \right)}&{{0_{1 \times p}}}
\end{bmatrix}}}^ \top }}_{{\eta _\varphi }\left( t \right)}{\rm{, }}\\
\end{array}
\end{equation*}
where
$\lambda _{{p} - 1}^{\top}\left( s \right) = {\begin{bmatrix}
{{s^{{p} - 1}}}& \cdots &s&1
\end{bmatrix}}$ and $\delta \left( s \right)$ denotes a monic Hurwitz polynomial of order ${p}$.

As $u\left(t\right)$ is an external reference, which generally independent from the measurement noise $\eta_{y}\left(t\right)$ and perturbation $f\left(t\right)$, then Assumption 2 is always met for linear closed loop dynamical system with the following matrices:
\begin{equation*}
{\mathcal{L}}^{\top}_{1}= {\begin{bmatrix}
{{0_{p \times p}}}&{{I_{p \times p}}}
\end{bmatrix}},\;\\{\mathcal{L}}^{\top}_{2}= {\begin{bmatrix}
{{I_{p \times p}}}&{{0_{p \times p}}}
\end{bmatrix}}.
\end{equation*}

Therefore, the  necessary condition $2m \ge n$ for convergence from Theorem 4 is also always satisfied ($m=p$ and $n=2p$). Moreover, using the results from \cite{Boyd} for the considered linear dynamical system, it seems to be possible to derive the constructive necessary and sufficient conditions to verify Assumption 3. Consequently, the proposed estimator \eqref{30}, \eqref{31}, \eqref{32} can be applied for parameter estimation of linear systems affected by perturbations and noise.$\hfill\blacksquare$

\section{Application to GPEBO Design \\ in Presence of Noise}

It is well-known \cite{3, 26} that the state estimation problem for a large class of systems can be transformed to the task of state reconstruction of the following state-affine nonlinear system
\begin{equation}\label{36}
\begin{array}{l}
\dot x\left( t \right) = {\mathbf{A}}\left( {u{\rm{,\;}}y{\rm{,\;}}t} \right)x\left( t \right) + {\bf{b}}\left( {u{\rm{,\;}}y{\rm{,\;}}t} \right){\rm{,\;}}x\left( {{t_0}} \right) = {x_0}{\rm{,}}\\
y\left( t \right) = {\mathbf{C}}\left( {u{\rm{,\;}}t} \right)x\left( t \right){\rm{,}}\\
\hat y\left( t \right) = y\left( t \right) + w\left( t \right){\rm{,}}
\end{array}
\end{equation}
where $x\left( t \right) \in {\mathbb{R}^n}$ is an unmeasured state, $u\left( t \right) \in \mathbb{R}$ denotes an input signal, $y\left( t \right) \in \mathbb{R}$ is a system noise-free output, $\hat y\left( t \right) \in \mathbb{R}$ is a measured output, which is corrupted by the bounded noise $w\left( t \right) \in \mathbb{R}$ and 
\begin{equation*}
A\left( t \right){\rm{:}} = {\mathbf{A}}\left( {u{\rm{,\;}}y{\rm{,\;}}t} \right){\rm{,\;}}b\left( t \right){\rm{:}} = {\rm{ }}{\mathbf{b}}\left( {u{\rm{,\;}}y{\rm{,\;}}t} \right){\rm{,\;}}C\left( t \right){\rm{:}} = {\mathbf{C}}\left( {u{\rm{,\;}}t} \right)
\end{equation*}
are known mappings.

The recently proposed \cite{4} GPEBO for the system \eqref{36} uses the noise-free output $y\left( t \right)$ that is generally unavailable for practical scenarios and includes dynamics
\begin{equation}\label{37}
\begin{array}{l}
\dot \xi \left( t \right) = A\left( t \right)\xi \left( t \right) + b\left( t \right){\rm{,\;}}\xi \left( {{t_0}} \right) = {0_n}{\rm{,}}\\
{{\dot \Phi }_A}\left( {t{\rm{,\;}}{t_0}} \right) = A\left( t \right){\Phi _A}\left( {t{\rm{,\;}}{t_0}} \right){\rm{,\;}}{\Phi _A}\left( {{t_0}{\rm{,\;}}{t_0}} \right) = {I_n}{\rm{,}}
\end{array}
\end{equation}
which provides that the state of the system \eqref{36} is obtained as
\begin{equation}\label{38}
x\left( t \right) = \xi \left( t \right) + {\Phi _A}\left( t \right)\theta {\rm{,}}
\end{equation}
where the unknown parameters $\theta  = {\xi _0}$ satisfy the LRE
\begin{equation}\label{39}
z\left( t \right) = {\varphi ^{\top}}\left( t \right)\theta
\end{equation}
with
\begin{equation*}
\begin{array}{c}
z\left( t \right){\rm{:}} = y\left( t \right) - C\left( t \right)\xi \left( t \right){\rm{,}}\\
{\varphi ^{\top}}\left( t \right){\rm{:}} = C\left( t \right){\Phi _A}\left( {t{\rm{,\;}}{t_0}} \right){\rm{.}}
\end{array}
\end{equation*}

Based on equations \eqref{38}, \eqref{39}, the GPEBO provides the following state observer
\begin{subequations}
\begin{equation}\label{40a}
\hat x\left( t \right) = \xi \left( t \right) + {\Phi _A}\left( {t{\rm{,\;}}{t_0}} \right)\hat \theta \left( t \right){\rm{,}}
\end{equation}
where the estimate $\hat \theta \left( t \right)$ is obtained via G + D estimator
\begin{equation}\label{40b}
\begin{array}{l}
\dot {\hat \theta}\!\left( t \right)\!=\!- \gamma {\rm{adj}}\left\{ {{I_n}\!-\!\Phi\!\left( t \right)} \right\}\left( {\left( {{I_n}\!-\!\Phi\!\left( t \right)} \right)\hat \theta \!\left( t \right)\!-\!Y\!\left( t \right)} \right)\!{\rm{,}}\\
\dot Y\left( t \right) = \Gamma \varphi \left( t \right)\left( {z\left( t \right) - {\varphi ^{\top}}\left( t \right)Y\left( t \right)} \right){\rm{,\;}}Y\left( {{t_0}} \right) = {0_n}{\rm{,}}\\
\dot \Phi \left( t \right) =  - \Gamma \varphi \left( t \right){\varphi ^{\top}}\left( t \right)\Phi \left( t \right){\rm{,\;}}\Phi \left( {{t_0}} \right) = {I_n}{\rm{,}}
\end{array}
\end{equation}
\end{subequations}
and $\hat \theta \left( {{t_0}} \right) = {{\hat \theta }_0}$, $\Gamma  = {\Gamma ^{\top}} > 0,{\rm{\;}}\gamma  > {\rm{0}}{\rm{.}}$

Recently it has been shown \cite{4} that the GPEBO with G+D interlaced parameter estimator ensures globally exponentially stable (GES) observation of the system \eqref{36} state iff the system \eqref{36} is observable. However, in the presence of measurement noise, the dynamics \eqref{37} cannot be implemented, the regressor $\varphi \left( t \right)$ and regressand $z\left( t \right)$ are unmeasurable and, therefore, GPEBO augmented with G+D estimator provides only convergence of the state estimation error to some residual set with a bound, which cannot be reduced arbitrarily. Therefore, the actual task in the direction of GPEBO improvement is stated as follows.

\textbf{Observation Problem in Presence of Noise:} Using the output $\hat y\left( t \right)$ corrupted by noise, the goal is to design a state observer such that, for any ${\varepsilon _x} > 0$, the ball
\begin{equation*}
{{\cal E}_x}{\rm{:}} = \left\{ {\tilde x \in {\mathbb{R}^n}{\rm{:\;}}\left\| {\tilde x} \right\| \le {\varepsilon _x}} \right\}
\end{equation*}
is GES, \emph{i.e.} for any ${\varepsilon _x} > 0$ and every ${\rho _x} > 0$, there exist $m\left( {{t_0}{\rm{,\;}}{\rho _x}} \right) > 0$ and $\lambda \left( {{t_0}{\rm{,\;}}{\rho _x}} \right) > 0$ such that, if $\left\| {\tilde x\left( {{t_0}} \right)} \right\| \le {\rho _x}$, then the state estimation error $\tilde x\left( t \right) = \hat x\left( t \right) - x\left( t \right)$ for all $t \ge {t_0}$ satisfies the bound
 \begin{equation*}
\quad\quad\quad\left\| {\tilde x\left( t \right)} \right\| \le m\left( {{t_0}{\rm{,\;}}{\rho _x}} \right){e^{ - \lambda \left( {{t_0}{\rm{,\;}}{\rho _x}} \right)\left( {t - {t_0}} \right)}}\left\| {\tilde x\left( {{t_0}} \right)} \right\| + {\varepsilon _x}.\quad\blacksquare
\end{equation*}

From the direct inspection of the GPEBO design procedure \eqref{37}-\eqref{40a} and the results of Theorems 1-4, it is obvious that the above-mentioned observation problem can be solved via simple application, instead of G+D estimator \eqref{40b}, of any of the proposed estimators \eqref{10}, \eqref{23}, \eqref{30} with properly defined assumptions with respect to the system observability Gramian, measurement noise and initial conditions. In this study, in order to satisfy space limitations, we combine the GPEBO observer \eqref{40a} only with a very simple estimation law \eqref{10} and, therefore, introduce the following assumptions.

\begin{assumption} The arguments of the system \eqref{36} matrices are input and time only, i.e.
\begin{equation*}
{\textstyle{{\partial {\mathbf{A}}\left( {u{\rm{,\;}}y{\rm{,\;}}t} \right)} \over {\partial y}}} = 0,{\rm{\;}}{\textstyle{{\partial {\mathbf{b}}\left( {u{\rm{,\;}}y{\rm{,\;}}t} \right)} \over {\partial y}}} = 0{\rm{\;for\;all\;}}u\left( t \right) \in \mathbb{R}{\rm{\;and\;}}t \ge {t_0}
\end{equation*}
and, therefore, ${\mathbf{A}}\left( {u{\rm{,\;}}y{\rm{,\;}}t} \right){\rm{:}}\!=\!{\mathbf{A}}\left( {u{\rm{,\;}}t} \right){\rm{,\;}}{\mathbf{b}}\left( {u{\rm{,\;}}y{\rm{,\;}}t} \right){\rm{:}}\!=\!{\mathbf{b}}\left( {u{\rm{,\;}}t} \right)$. 
\end{assumption}
\begin{assumption} The state transition matrix of the homogeneous part of the system \eqref{36} for some ${\overline \Phi _A} > 0$ satisfies
\begin{equation*} 
\left\| {{\Phi _A}\left( {t{\rm{,\;}}\tau } \right)} \right\| \le {\overline \Phi _A}{\rm{\;for\;all\;}}t \ge \tau  \ge {t_0}.
\end{equation*}
\end{assumption}

\begin{assumption} For the observability Gramian 
\begin{equation*}
\begin{array}{l}
{\cal O}\left( {{t_1}{\rm{,\;}}{t_2}} \right){\rm{:}} = \int\limits_{{t_1}}^{{t_2}} {\varphi \left( \tau  \right){\varphi ^{\top}}\left( \tau  \right)d\tau },
\end{array}
\end{equation*}
the following conditions are met:
\begin{enumerate}
\item[{\textbf{C1)}}] there exist $\overline{\alpha}  \ge  \underline{\alpha}> 0,{\rm{\;}}T > 0$ such that 
\begin{equation*}
\overline{\alpha} I_{n} \ge {\cal O}\left( {t{\rm{,\;}}t + T} \right) \ge \underline{\alpha}I_{n}    
\end{equation*}
  for all $t \ge {t_0}$, i.e. $\varphi  \in {\rm{PE}}$ (the system \eqref{36} is Uniformly Completely Observable (UCO)),
\item[{\textbf{C2)}}] there exists $\Lambda  \in {\mathbb{R}^{n \times n}}$, such that {\begin{equation*}
\mathop {{\rm{lim}}}\limits_{T \to \infty } \mathop {{\rm{sup}}}\limits_{t \geqslant {t_0}} \frac{1}{T}\left\| {{\cal{O}}\left( {t,\;t + T} \right) - \Lambda } \right\| = 0,
\end{equation*}} i.e. $\varphi  \in {\rm{ST}}$.  
\end{enumerate}
\end{assumption}
\begin{assumption}$\varphi  \bot w.$ 
\end{assumption}

If Assumption 5 is met, then the signals $\xi \left( t \right)$ and ${\Phi _A}\left( {t{\rm{,\;}}{t_0}} \right)$ are both measurable for all $t \ge {t_0}$, and, therefore, the LRE \eqref{3} with the following regressor and regressand can be obtained:
\begin{equation}\label{41}
\begin{array}{l}
{{\hat \varphi }^{\top}}\left( t \right){\rm{:}} = {\varphi ^{\top}}\left( t \right) = C\left( t \right){\Phi _A}\left( {t{\rm{,\;}}{t_0}} \right){\rm{,}}\\
\hat z\left( t \right){\rm{:}} = \hat y\left( t \right) - C\left( t \right)\xi \left( t \right).
\end{array}
\end{equation}

Consequently, as Assumptions 7 and 8 are met, then the estimation law \eqref{10} can be applied to estimate the unknown initial conditions $\theta$, which are required for application of observer \eqref{40a}.
\begin{theorem}
    The following state observer provides a solution to the stated Observation Problem in Presence of Noise
\begin{equation}\label{42}
\hat x\left( t \right) = \xi \left( t \right) + {\Phi _A}\left( {t{\rm{,\;}}{t_0}} \right)\hat \theta \left( t \right){\rm{,}}
\end{equation}
\begin{equation}\label{43}
\begin{array}{l}
\dot {\hat \theta} \left( t \right) =  - \gamma {\rm{adj}}\left\{ {\Phi \left( t \right)} \right\}\left( {\Phi \left( t \right)\hat \theta \left( t \right) - Y\left( t \right)} \right){\rm{,}}\\
\dot Y\left( t \right) = {\textstyle{1 \over T}}\left[ {\hat \varphi \left( t \right)\hat z\left( t \right) - \hat \varphi \left( {t - T} \right)\hat z\left( {t - T} \right)} \right]{\rm{,\;}}\\
\dot \Phi \left( t \right) = {\textstyle{1 \over T}}\left[ {\hat \varphi \left( t \right){{\hat \varphi }^{\top}}\left( t \right) - \hat \varphi \left( {t - T} \right){{\hat \varphi }^{\top}}\left( {t - T} \right)} \right]{\rm{,}}
\end{array}
\end{equation}	
where $\hat \theta \left( {{t_0}} \right) = {{\hat \theta }_0}{\rm{,}}$ $Y\left( {{t_0}} \right) = {0_n}{\rm{,}}$ $\Phi \left( {{t_0}} \right) = {0_{n \times n}}{\rm{,}}$ $\gamma  > 0$ and $T > 0$.
\end{theorem}
\begin{proof}Proof is presented in Supplementary material \cite{27}.\end{proof}

Consequently, in comparison with the baseline G + D based GPEBO \cite{4}, the proposed observer \eqref{42}, \eqref{43} has the following main distinguishing features: 1) the observer \eqref{42}, \eqref{43} provided GES with arbitrarily small compact set in case when the measured output signal is affected by noise, 2) the observer \eqref{42}, \eqref{43} requires to meet the UCO/PE assumption (\textbf{C1} from Assumption 7), but the G+D based GPEBO requires only the system observability (strictly weaker assumption). The first distinguishing feature of the proposed observer is its advantage and novelty, but the second feature is its serious drawback and limitation. However, to the best of authors’ knowledge, in existing literature, there are no observers, which provide solution of the stated observation problem in the presence of noise even in UCO case, which slightly mitigates such mentioned drawback.
\begin{remark}
    Violation of Assumption 5 does not allow one to propagate the smallness of the parametric error $\tilde \theta \left( t \right)$ to the smallness of the estimation error $\tilde x\left( t \right)$. In this case, the error $\tilde x\left( t \right)$ explicitly depends from the measurement noise $w\left( t \right)$. Therefore, if Assumption 5 is not met, then the observer \eqref{42}, \eqref{43} ensures convergence of the parametric error $\tilde \theta \left( t \right)$ to the neighborhood of zero with an adjustable bound, while for the error $\tilde x\left( t \right)$ – to the neighborhood of zero, which cannot be reduced arbitrarily.
\end{remark}
\begin{remark}
    Note that, if the system initial condition ${\xi _0}$ lies on the manifold $\left\{ {\xi  \in {\mathbb{R}^n}{\rm{:\;}}{{\cal H}^{\top}}\xi  = {0_{n - m}}} \right\}$, then the estimator \eqref{23} also provides solution of the stated observation problem. 
\end{remark}
\section{Simulation}

Consider the system \eqref{36}, which is defined as follows:
\begin{equation*}
\begin{array}{c}
{\mathbf{A}}\left( {u{\rm{,\;}}t} \right){\rm{:}} = {\begin{bmatrix}
0&1\\
{ - {\omega ^2}}&0
\end{bmatrix}}{\rm{,\;}}{\mathbf{b}}\left( {u{\rm{,\;}}t} \right){\rm{:}} = {\begin{bmatrix}
0\\
{u\left( t \right)}
\end{bmatrix}}{\rm{,\;}}{\xi _0} = {\begin{bmatrix}
2\\
{ - 1}
\end{bmatrix}}{\rm{, }}\\
{\bf{C}}\left( {u{\rm{,\;}}t} \right) {\rm{:}} = {\begin{bmatrix}
1&0
\end{bmatrix}}{\rm{,\;}}u\left( t \right) = 1,{\rm{\;}}w\left( t \right) = 0.1{\rm{sin}}\left( {20t} \right){\rm{, }}
\end{array}
\end{equation*}
for which we have:
\begin{equation*}
{\hat \varphi ^{\top}}\left( t \right){\rm{:}} = {\begin{bmatrix}{{\rm{cos}}\left( {\omega t} \right)}&{{\textstyle{1 \over \omega }}{\rm{sin}}\left( {\omega t} \right)}
\end{bmatrix}}{\rm{,}}
\end{equation*}
and therefore, in case ${\cal L}_1^{\top} = {I_2}{\rm{,\;}}{\cal L}_2^{\top} = 0$ and $\omega  \ne 20$, it holds {uniformly in time $t$} that
\begin{equation*}
\begin{array}{l}
\mathop {{\rm{lim}}}\limits_{T \to \infty }\!\! {\textstyle{1 \over {{T}} }}\!{\begin{bmatrix}\!
{{\textstyle{{{\rm{sin}}\left( {2\omega \left( {T + t} \right)} \right) - {\rm{sin}}\left( {2\omega t} \right) + 2T\omega } \over {4\omega }}}}\!\!&\!\!{{\textstyle{{{\rm{cos}}\left( {2\omega t} \right)-{\rm{cos}}\left( {2\omega \left( {T + t} \right)} \right)} \over {4{\omega ^2}}}}}\\
{{\textstyle{{{\rm{cos}}\left( {2\omega t} \right)-{\rm{cos}}\left( {2\omega \left( {T + t} \right)} \right)} \over {4{\omega ^2}}}}}\!\!&\!\!{{\textstyle{{2T\omega  - {\rm{sin}}\left( {2\omega \left( {T + t} \right)} \right) + {\rm{sin}}\left( {2\omega t} \right)} \over {4{\omega ^3}}}}}\!
\end{bmatrix}}\\
 = {\begin{bmatrix}
{0.5}&0\\
0&{0.5{\omega ^{ - 2}}}
\end{bmatrix}}=\Lambda{\rm{,}}\\
\mathop {{\rm{lim}}}\limits_{T \to \infty } {\textstyle{1 \over {{T}} }}\int\limits_{t}^{t+T} {{\cal L}_1^{\top}\hat \varphi \left( \tau  \right)w\left( \tau  \right)d\tau }  = 0.
\end{array}
\end{equation*}

Then Assumptions 5-8 are met for the example under consideration. For better illustration of the properties of the proposed observer \eqref{42}, \eqref{43}, the measurement noise was defined as:
\begin{equation*}
w\left( t \right) = 0.1{\rm{sin}}\left( {20t} \right) + {\eta _y}\left( t \right){\rm{,}}
\end{equation*}
where ${\eta _y}\left( t \right)$ stands for a white noise generated in Simulink by ``Band Limited White Noise'' block with the parameters Noise Power = 0.1, Sample Time = 0.01, Seed = 23341.

For comparison purposes, we also implemented the G+D based GPEBO \cite{4} defined as follows:
\begin{equation}\label{46}
\hat x\left( t \right) = \xi \left( t \right) + {\Phi _A}\left( {t{\rm{,\;}}{t_0}} \right)\hat \theta \left( t \right){\rm{,}}
\end{equation}
\begin{equation}\label{47}
\begin{array}{l}
\dot {\hat \theta}\!\left( t \right)\!=\!- \gamma {\rm{adj}}\left\{ {{I_2}\!-\!\Phi\!\left( t \right)} \right\}\left( {\left( {{I_2}\!-\!\Phi\!\left( t \right)} \right)\hat \theta \!\left( t \right)\!-\!Y\!\left( t \right)} \right)\!{\rm{,}}\\
\dot Y\left( t \right) = \Gamma \hat \varphi \left( t \right)\left( {z\left( t \right) - {{\hat \varphi }^{\top}}\left( t \right)Y\left( t \right)} \right){\rm{,}}\\
\dot \Phi \left( t \right) =  - \Gamma \hat \varphi \left( t \right){{\hat \varphi }^{\top}}\left( t \right)\Phi \left( t \right){\rm{,}}
\end{array}
\end{equation}
where $\hat \theta \left( {{t_0}} \right) = {{\hat \theta }_0}{\rm{,}}$ $Y\left( {{t_0}} \right) = {0_2}{\rm{,}}$ $\Phi \left( {{t_0}} \right) = {I_2}$.
	
The parameters of the observers \eqref{42}, \eqref{43} and \eqref{46}, \eqref{47} were chosen as:
\begin{equation*}
\Gamma  = {I_2},\;T=36,\;\gamma  =  \begin{cases}
100,{\rm{\;for\;}}\eqref{43}\\
{\rm{1}}{\rm{,\;for\;}}\eqref{47}
\end{cases}
\end{equation*}

The transients for ${\rm{ln}}\left( {\left\| {\tilde x} \right\|} \right)$ and  ${\tilde \theta _i}\left( t \right)$  are presented respectively in Fig.1 and Fig.2.
\begin{figure}[!ht]
\begin{center}
\includegraphics[scale=0.51]{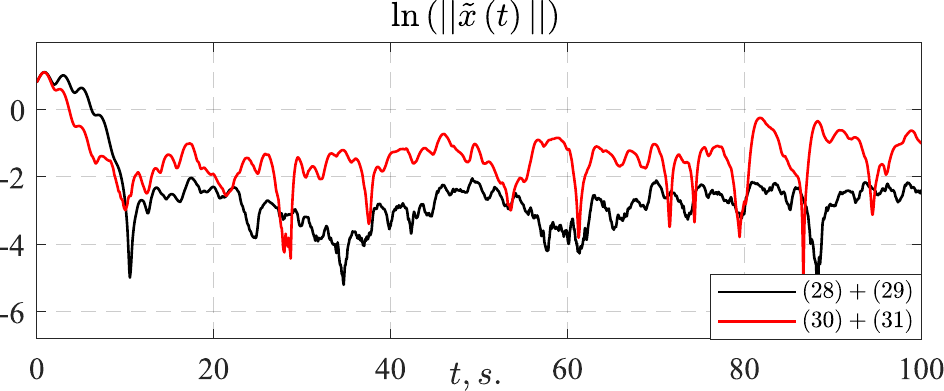}
\caption{Behavior of ${\rm{ln}}\left( {\left\| {\tilde x} \right\|} \right)$  for \eqref{42}+\eqref{43} and \eqref{46}+\eqref{47}.} 
\end{center}
\end{figure}

\begin{figure}[!ht]
\begin{center}
\includegraphics[scale=0.45]{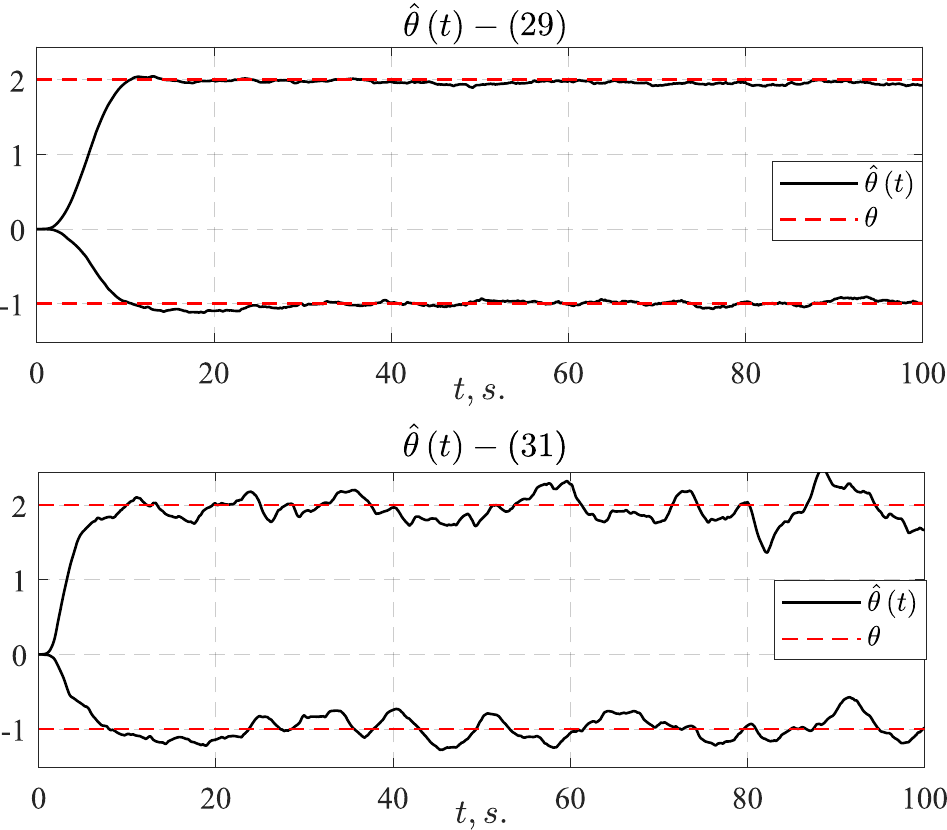}
\caption{Behavior of  ${\tilde \theta _i}\left( t \right)$ for \eqref{43} and \eqref{47}.} 
\end{center}
\end{figure}

The simulation results illustrate the theoretical conclusions of Theorems 1 and 5. In comparison with the observer \eqref{46}, \eqref{47}, the proposed one \eqref{42}, \eqref{43} ensures GES with a compact set with an {arbitrarily} small and adjustable bound.

\section{Conclusion and Discussion}
The paper proposes three new estimation laws \eqref{10}, \eqref{23}, \eqref{30}, each of which ensures exponential convergence of the parametric error to an arbitrarily small neighborhood of zero in case the regression equation is affected by an additive perturbation that can be caused by measurement noise (that corrupts regressor and regressand), as well as external perturbations. The law \eqref{30} has the weakest convergence conditions, which ensures convergence of the parametric error to an arbitrarily small neighborhood of zero in case more than a half of the regressor elements are independent (in the sense of Definition 3) from additive perturbation. 

In general, further research scope includes the following main directions:
\begin{enumerate}
\item[\textbf{FR1)}] application of the proposed estimation algorithms to solve adaptive control and state estimation problems, for example, as a part of composite adaptive laws \cite{16},
\item[\textbf{FR2)}] derivation of the necessary and sufficient conditions to meet the implication $\hat \varphi  \in {\rm{PE}} \Rightarrow \phi  \in {\rm{PE}}$ from Assumption 4,
\item[\textbf{FR3)}] improvement of the obtained results and relaxation of the adopted assumptions.
\end{enumerate}

\end{document}